\documentclass[a4paper,11pt]{article}
 
\usepackage{microtype}

\usepackage{tikz}
\usetikzlibrary{arrows,chains,matrix,positioning,scopes}
\makeatletter
\tikzset{join/.code=\tikzset{after node path={%
\ifx\tikzchainprevious\pgfutil@empty\else(\tikzchainprevious)%
edge[every join]#1(\tikzchaincurrent)\fi}}}
\makeatother
\tikzset{>=stealth',every on chain/.append style={join},
         every join/.style={->}}
\tikzstyle{labeled}=[execute at begin node=$\scriptstyle,
execute at end node=$]
\tikzstyle{vertex}=[circle, draw, fill=blue!20, inner sep=0pt, minimum width=15pt]
\tikzstyle{finalvertex}=[circle, draw, fill=red!30, inner sep=0pt, minimum width=15pt]
\tikzstyle{initialvertex}=[circle, draw, fill=green!30, inner sep=0pt, minimum width=15pt]
\tikzstyle{en}=[text=blue]
\title{Small-depth Multilinear Formula Lower Bounds for Iterated Matrix Multiplication, with Applications}

\author{Suryajith Chillara\\
  Department of CSE, IIT Bombay.\\  \texttt{suryajith@cse.iitb.ac.in}
  \and
  Nutan Limaye\\
  Department of CSE, IIT Bombay.\\ \texttt{nutan@cse.iitb.ac.in}
  \and
  Srikanth Srinivasan\\
  Department of Mathematics, IIT Bombay.\\
  \texttt{srikanth@math.iitb.ac.in}}

\usepackage{amsmath}
\usepackage{amssymb}
\usepackage{amsthm}
\usepackage{fullpage}
\usepackage[colorlinks]{hyperref}
\hypersetup{
  linkcolor=[rgb]{0.3,0.3,0.6},
  citecolor=[rgb]{0.2, 0.6, 0.2},
  urlcolor=[rgb]{0.6, 0.2, 0.2}
}

\usepackage{algorithm,algorithmic}

\newtheorem{theorem}{Theorem}
\newtheorem{corollary}[theorem]{Corollary}
\newtheorem{lemma}[theorem]{Lemma}
\newtheorem{observation}[theorem]{Observation}
\newtheorem{proposition}[theorem]{Proposition}
\newtheorem{definition}[theorem]{Definition}
\newtheorem{claim}[theorem]{Claim}

\newcommand{\prob}[2]{\mathop{\mathrm{Pr}}_{#1}[#2]}
\newcommand{\avg}[2]{\mathop{\textbf{E}}_{#1}[#2]}
\newcommand{\poly}{\mathop{\mathrm{poly}}}

\newcommand{\abs}[1]{\left|#1\right|}

\newcommand{\mc}[1]{\mathcal{#1}}
\newcommand{\rank}{\mathrm{rank}}

\newtheorem*{rep@theorem}{\rep@title}
\newcommand{\newreptheorem}[2]{%
\newenvironment{rep#1}[1]{%
 \def\rep@title{#2 \ref{##1}}%
 \begin{rep@theorem}}%
 {\end{rep@theorem}}}
\makeatother

\newreptheorem{theorem}{Theorem}
\newreptheorem{lemma}{Lemma}

\newcommand{\F}{\mathbb{F}}
\newcommand{\IMM}{\mathrm{IMM}}
\newcommand{\Vars}{\mathrm{Vars}}
\newcommand{\vars}{\mathrm{Vars}}
\newcommand{\supp}{\mathrm{Supp}}
\newcommand{\M}{{M}}
\newcommand{\Img}{\mathrm{Img}}

\newcommand{\simple}{simple}

\usepackage[colorinlistoftodos]{todonotes}

\begin{document}

\maketitle
\begin{abstract}
The complexity of Iterated Matrix Multiplication is a central theme in Computational Complexity theory, as the problem is closely related to the problem of separating various complexity classes within $\mathrm{P}$. In this paper, we study the algebraic formula complexity of multiplying $d$ many $2\times 2$ matrices, denoted  $\IMM_{d}$, and show that the well-known divide-and-conquer algorithm cannot be significantly improved at any depth, as long as the formulas are multilinear.

Formally, for each depth $\Delta \leq \log d$, we show that any product-depth $\Delta$ multilinear formula for $\IMM_d$ must have size $\exp(\Omega(\Delta d^{1/\Delta})).$ It also follows from this that any multilinear circuit of product-depth $\Delta$ for the same polynomial of the above form must have a size of $\exp(\Omega(d^{1/\Delta})).$ In particular, any polynomial-sized multilinear formula for $\IMM_d$ must have depth $\Omega(\log d)$, and any polynomial-sized multilinear circuit for $\IMM_d$ must have depth $\Omega(\log d/\log \log d).$ Both these bounds are tight up to constant factors.

Our lower bound has the following consequences for multilinear formula complexity. 

\begin{enumerate}
\item \textbf{Depth-reduction}: A well-known result of Brent (JACM 1974) implies that any formula of size $s$ can be converted to one of size $s^{O(1)}$ and depth $O(\log s)$; further, this reduction continues to hold for multilinear formulas. On the other hand, our lower bound implies that any depth-reduction in the multilinear setting cannot reduce the depth to $o(\log s)$ without a superpolynomial blow-up in size.

\item \textbf{Separations from general formulas}: Shpilka and Yehudayoff (FnTTCS 2010) asked whether general formulas can be more efficient than multilinear formulas for computing multilinear polynomials. Our result, along with a non-trivial upper bound for $\IMM_{d}$ implied by a result of Gupta, Kamath, Kayal and Saptharishi (SICOMP 2016), shows that for any size $s$ and product-depth $\Delta = o(\log s),$ general formulas of size $s$ and product-depth $\Delta$ cannot be converted to multilinear formulas of size $s^{\omega(1)}$ and product-depth $\Delta,$ when the underlying field has characteristic zero.
\end{enumerate}

\end{abstract}

\section{Introduction}
Algebraic Complexity theory is the study of the complexity of those computational problems that can be phrased as computing a multivariate polynomial $f(x_1,\ldots,x_N)\in \mathbb{F}[x_1,\ldots, x_N]$ over elements $x_1,\ldots,x_N\in \F.$ Many central algorithmic problems such as the Determinant, Permanent, Matrix product etc. can be cast in this framework. The natural computational models that we consider in this setting are models such as \emph{Algebraic circuits}, \emph{Algebraic Branching Programs} (ABPs), and \emph{Algebraic formulas} (or just formulas), all of which use the natural algebraic operations of $\F[x_1,\ldots,x_N]$ to compute the polynomial $f$. These models have by now been the subject of a large body of work with many interesting upper bounds (i.e. circuit constructions) as well as lower bounds (i.e. impossibility results). (See, e.g. the surveys~\cite{sy,github} for an overview of many of these results.)

Despite this, many fundamental questions remain unresolved. An important example of such a question is that of proving lower bounds on the size of formulas for the \emph{Iterated Matrix Multiplication} problem, which is defined as follows. Given $d$ $n\times n$ matrices $M_1,\ldots,M_d$, we are required to compute (an entry of) the product $M_1\cdots M_d$; we refer to this problem as $\IMM_{n,d}.$ Proving superpolynomial lower bounds on the size of formulas for this problem is equivalent to separating the power of polynomial-sized ABPs from polynomial-sized formulas, which is the algebraic analogue of separating the Boolean complexity classes $\mathrm{NL}$ and $\mathrm{NC}^1.$

A standard divide-and-conquer algorithm yields the best-known formulas for $\IMM_{n,d}.$ More precisely, for any $\Delta \leq \log d$, this approach yields a formula of product-depth\footnote{The \emph{product-depth} of an arithmetic circuit or formula is the maximum number of product gates on a path from output to input. If the product-depth of a circuit or formula is $\Delta$, then its depth can be assumed to be at least $2\Delta-1$ and at most $2\Delta +1.$} $\Delta$ and size $n^{O(\Delta d^{1/\Delta})}$ for $\IMM_{n,d}$ and choosing $\Delta = \log d$ yields the current best formula upper bound of $n^{O(\log d)}$, which has not been improved in quite some time. On the other hand, separating the power of ABPs and formulas is equivalent to showing that $\IMM_{n,d}$ does not have formulas of size $\poly(nd).$

The Iterated Matrix Multiplication problem has many nice features that render its complexity an interesting object to study. For one, it is the algebraic analogue of the Boolean reachability problem, and thus any improved formula upper bounds for $\IMM_{n,d}$ could lead to improved Boolean circuit upper bounds for the reachability problem, which would resolve a long-standing open problem in that area. For another, this problem has strong self-reducibility properties, which imply that improving on the simple divide-and-conquer approach to obtain formulas of size $n^{o(\log d)}$ for any $d$ would lead to improved upper bounds for all $D > d$; this implies that the lower-degree variant is no easier than the higher-degree version of the problem, which can be very useful (e.g. for homogenization~\cite{Raz13}). Finally, the connection to the Reachability problem imbues $\IMM_{n,d}$ with a rich combinatorial structure via its graph theoretic interpretation, which has been used extensively in lower bounds for depth-$4$ arithmetic circuits~\cite{FLMS,KLSS,KS14,KNS16,KST16}. 

We study the formula complexity of this problem in the \emph{multilinear} setting, which restricts the underlying formulas to only compute multilinear polynomials at intermediate stages of computation. Starting with the breakthrough work of Raz~\cite{Raz}, many lower bounds have been proved for multilinear models of computation~\cite{ry08,ry09,RSY08,DMPY12}. Further, it is known by a result of Dvir, Malod, Perifel and Yehudayoff~\cite{DMPY12} that multilinear ABPs are in fact superpolynomially more powerful than multilinear formulas. Unfortunately, however, this does not imply any non-trivial lower bound for Iterated Matrix Multiplication (see the Related Work section below), and as far as we know, it could well be the case that there are multilinear formulas that beat the divide-and-conquer approach in computing this polynomial.

Here, we are able to show that this is not the case for the problem of multiplying $2\times 2$ matrices (and by extension $c\times c$ matrices for any constant $c$) at any product-depth. Our main theorem is the following.

\begin{theorem}
\label{thm:mainlbd-informal}
For $\Delta\leq \log d$, any product-depth $\Delta$ multilinear formula that computes $\IMM_{2,d}$ must have size $2^{\Omega(\Delta d^{1/\Delta})}$.
\end{theorem}

This lower bound strengthens a result of Nisan and Wigderson~\cite{nw1997} who prove a similar lower bound in the more restricted \emph{set-multilinear} setting. 

Our result is also qualitatively different from the previous lower bounds for multilinear formulas since $\IMM_{2,d}$ does in fact have polynomial-sized formulas of product-depth $O(\log d)$ (via the divide-and-conquer approach), whereas we show a superpolynomial lower bound for product-depth $o(\log d)$. This observation leads to interesting consequences for multilinear formula complexity in general, which we now describe. 

\subparagraph*{Depth Reduction.}  An important theme in Circuit complexity is the interplay between the size of a formula or circuit and its depth~\cite{brent,spira,VSBR,av,Tav13}. In the context of algebraic formulas, a result of Brent~\cite{brent} says that any formula of size $s$ can be converted into another of size $s^{O(1)}$ and depth $O(\log s).$ Further, the proof of this result also yields the same statement for multilinear formulas. 

Can the result of Brent be improved? Theorem~\ref{thm:mainlbd-informal} implies that the answer is no in the multilinear setting. More precisely, since the $\IMM_{2,d}$ polynomial (over $O(d)$ variables) has formulas of size $\poly(d)$ and depth $O(\log d)$ but no formulas of size $d^{O(1)}$ and depth $o(\log d)$ (by Theorem~\ref{thm:mainlbd-informal}), we see that any multilinear depth-reduction procedure that reduces the depth of a size-$s$ formula to $o(\log s)$ must incur a superpolynomial blow-up in size. This strengthens a result of Raz and Yehudayoff~\cite{ry09}, whose results imply that any depth-reduction of multilinear formulas to depth $o(\sqrt{\log s}/\log \log s)$ should incur a superpolynomial blow-up in size. It is also an analogue in the algebraic setting of some recent results proved for Boolean circuits~\cite{Ross15, RS17}.

\subparagraph*{Multilinear vs. general formulas.} Shpilka and Yehudayoff~\cite{sy} ask the question of whether general formulas can be more efficient at computing multilinear polynomials than multilinear formulas. This is an important question, since we have techniques for proving lower bounds for multilinear formulas, whereas the same question for general formulas (or even depth-$3$ formulas over large fields) remains wide open. 

We are able to make progress towards this question here by showing a separation between the two models for small depths when the underlying field has characteristic zero. We do this by using Theorem~\ref{thm:mainlbd-informal} in conjunction with a (non-multilinear) formula \emph{upper bound} for $\IMM_{2,d}$ over fields of characteristic zero due to Gupta et al.~\cite{GKKSdepth3}. In particular, the result of Gupta et al.~\cite{GKKSdepth3} implies that for any depth $\Delta,$ the polynomial $\IMM_{2,d}$ has formulas of product depth $\Delta$ and size $2^{O(\Delta d^{1/2\Delta})},$ which is considerably smaller than our lower bound in the multilinear case for small $\Delta.$ From this, it follows that for any size parameter $s$ and product-depth $\Delta = o(\log s),$ general formulas of size $s$ and product-depth $\Delta$ cannot be converted to multilinear formulas of size $s^{\omega(1)}$ and product-depth $\Delta$. Improving our result to allow for $\Delta = O(\log s)$ would resolve the question entirely.

\subparagraph*{Related Work.} The multilinear formula model has been the focus of a large body of work on Algebraic circuit lower bounds. Nisan and Wigderson~\cite{nw1997} proved some of the early results in this model by showing size lower bounds for small-depth \emph{set-multilinear}\footnote{Set-multilinear circuits are further restrictions of multilinear circuits. A set-multilinear circuit for $\IMM_{n,d}$ is defined by the property that each intermediate polynomial computed must be a linear combination of monomials that contain exactly one variable from each matrix $M_i$ ($i\in S$), for some choice of $S \subseteq [d].$} circuits computing $\IMM_{2,d}$.  They showed that any product-depth $\Delta$ circuit for $\IMM_{2,d}$ must have a size of $2^{\Omega(d^{1/\Delta})}$ matching the upper bound from the divide-and-conquer algorithm for $\Delta = o(\log d/\log \log d)$. Our lower bounds for multilinear formulas imply similar lower bounds for multilinear circuits of product-depth $\Delta$.

Raz~\cite{Raz} proved the first superpolynomial lower bound for multilinear formulas by showing an $n^{\Omega(\log n)}$ lower bound for the $n\times n$ Determinant and Permanent polynomials. This was further strengthened by the results of Raz~\cite{Raznc2nc1} and Raz and Yehudayoff~\cite{ry08} to a similar lower bound for an explicit polynomial family that has polynomial-sized multilinear \emph{circuits}. In particular, these results show the tightness of the depth-reduction procedure for algebraic \emph{circuits} in the multilinear setting~\cite{VSBR,ry08}.

Similar polynomial families were also used in the work of Raz and Yehudayoff~\cite{ry09} to prove \emph{exponential} lower bounds for multilinear constant-depth circuits. By proving a tight lower bound for depth-$\Delta$ circuits computing an explicit polynomial (similar to the construction of Raz~\cite{Raznc2nc1}), Raz and Yehudayoff~\cite{ry09} showed superpolynomial separations between  multilinear circuits of different depths. 

In particular, the result of Raz and Yehudayoff~\cite{ry09} implies that the polynomial families of~\cite{Raznc2nc1,ry08}, which have formulas of size $n^{O(\log n)}$, cannot be computed by formulas of size less than some $s(n) = n^{\omega(\log n)}$ if the product-depth $\Delta = o(\log n/\log \log n).$ This yields the superpolynomial separation between formulas of size $s$ and depth $o(\sqrt{\log s}/\log \log s)$ alluded to above. Unfortunately, these polynomials also have nearly optimal formulas of depth just $O(\log n) = O(\sqrt{\log s})$, so they cannot be used to obtain the optimal size $s$ vs depth $o(\log s)$ separation we obtain here.

Dvir et al.~\cite{DMPY12} showed that there is an explicit polynomial on $n$ variables that has multilinear ABPs of size $\poly(n)$ but no multilinear formulas of size less than $n^{\Omega(\log n)}.$ One might hope that this yields a superpolynomial lower bound for multilinear formulas computing $\IMM_{N,d}$ for some $N,d$ but this unfortunately does not seem to be the case. The reason for this is that while any polynomial $f$ on $n$ variables that has an ABP of size $\poly(n)$ can be reduced via variable substitutions to $\IMM_{N,d}$ for $N,d = n^{O(1)}$, this reduction might substitute different variables in the $\IMM_{N,d}$ polynomial by the same variable $x$ of $f$ and in the process destroy multilinearity. 

Gupta et al.~\cite{GKKSdepth3} showed the surprising result that general (i.e. non-multilinear) formulas of depth-$3$ can beat the divide-and-conquer approach for computing $\IMM_{n,d},$ when the underlying field has characteristic zero. Their result implies that, in this setting, $\IMM_{n,d}$ has product-depth $1$ formulas of size $n^{O(\sqrt{d})}$, as opposed to the $n^{O(d)}$-sized formula that is obtained from the traditional divide-and-conquer approach.  Using the self-reduction properties of $\IMM_{n,d}$, this can be easily seen to imply the existence of $n^{O(\Delta d^{1/2\Delta})}$-sized formulas of product-depth $\Delta.$ This construction uses the fact that the formulas are allowed to be non-multilinear. Our result shows that this cannot be avoided.

\subparagraph*{Proof Overview.} The proof follows a two-step process as in ~\cite{sy,DMPY12}.

The first step is a ``product lemma'' where we show that any multilinear polynomial $f$ on $n$ variables that has a small multilinear formula can also be computed as a sum of a small number of polynomials each of which is a product of many polynomials on disjoint sets of variables; if such a term is the product of $t$ polynomials, we call it a $t$-product polynomial.\footnote{The polynomials in our decomposition can also have a different form which we choose to ignore for now.} It is known~\cite[Lemma~3.5]{sy} that if $f$ has a formula of size $s$, then we can ensure a decomposition into a sum of at most $s$ many $\Omega(\log n)$-product polynomials. We show that if the formula further is known to have depth $\Delta$ then the number of factors can be increased to $\Omega(\Delta n^{1/\Delta})$. In particular, note that this is $\omega(\log n)$ as long as $\Delta = o(\log n)$: this allows us to obtain superpolynomial lower bounds for up to this range of parameters. 

Similar lemmas were already known in the small-depth setting~\cite{ry09}, but they do not achieve the parameters of our lemma here. However, the lemma of~\cite{ry09} satisfies the additional condition that every factor of each $t$-product polynomial in the decomposition depends on a ``large'' number of variables. Here, we only get that each factor depends on a non-zero number of variables, but this is sufficient to prove the lower bound we want. 

The second step is to use this decomposition to prove a lower bound. Specifically, we would like to say that the polynomial $\IMM_{2,d}$ has no small decomposition into terms of the above form. This is via a rank argument as in Raz~\cite{Raz}. Specifically, we partition the variables $X$ in our polynomial into two sets $Y$ and $Z$ and consider any polynomial $f(X)$ as a polynomial in the variables in $Y$ with coefficients from $\F[Z].$ The dimension of the space of coefficients (as vectors over the base field $\F$) is considered a measure of the complexity of $f$. 

It is easy to come up with a partition of the underlying variable set $X$ into $Y,Z$ so that the complexity of $\IMM_{2,d}$ is as large as possible. Unfortunately, we also have simple multilinear formulas that have maximum dimension w.r.t. this partition. Hence, this notion of complexity is not by itself sufficient to prove a lower bound. At this point, we follow an idea of Raz~\cite{Raz} and show something stronger for $\IMM_{2,d}$: we show that its complexity is quite \emph{robust} in the sense that it is full rank w.r.t. many different partitions.

More precisely, we carefully design a large space of \emph{restrictions} $\rho:X\rightarrow Y\cup Z\cup \F$ such that for any restriction $\rho,$ the resulting substitution of $\IMM_{2,d}$ continues to have high complexity w.r.t. the measure defined above. These restrictions are motivated by the combinatorial structure of the underlying polynomial, specifically the connection to Graph Reachability. 

The last step is to show that, for any $t$-product polynomial $f$, a random restriction from the above space of restrictions transforms it with high probability into a polynomial whose measure is small. Once we have this result, it follows that given a small multilinear formula, there is a restriction that transforms each term in its decomposition (obtained from the product lemma) into a small complexity polynomial. The subadditivity of rank then shows that the entire formula now has small complexity, and hence it cannot be computing $\IMM_{2,d}$ which by the choice of our restriction has high complexity.

\section{Preliminaries}

\subsection{Basic setup}
\label{sec:basic}

Unless otherwise stated, let $\F$ be an arbitrary field. Let $d\in \mathbb{N}$ a growing integer parameter. We define $X^{(1)},\ldots,X^{(d)}$ to be disjoint sets of variables where each $X^{(i)} = \{x^{(i)}_{j,k}\ |\ j,k\in [2]\}$ is a set of four variables that we think of forming a $2\times 2$ matrix. Let $X = \bigcup_{i\in [d]}X^{(i)}$.

A polynomial $P\in \F[X]$ is called \emph{multilinear} if the degree of $P$ in each variable $x\in X$ is at most $1$. 
We define the multilinear polynomial $\IMM_d \in \F[X]$ as follows. Consider the matrices $M^{(1)},\ldots,M^{(d)}$ where the entries of $M^{(i)}$ are the variables of $X^{(i)}$ arranged in the obvious way. Define the matrix $M = M^{(1)}\cdots M^{(d)}$; the entries of $M$ are multilinear polynomials over the variables in $X$. We define 
$$\IMM_d = M(1,1) + M(1,2),$$ 
i.e. the sum of the $(1,1)$th and $(1,2)$th entries of $M$. Note, in particular, that the polynomial $\IMM_d$ does not depend on the variables $x^{(1)}_{2,1}$ and $x^{(1)}_{2,2}$. 

This is a slight variant of the Iterated Matrix Multiplication polynomial seen in the literature, as it is usually defined to be either the matrix entry $M({1,1})$ or the trace $M({1,1}) + M({2,2})$. Our results can easily be seen to hold for these variants, but we deal with the definition above for some technical simplicity.

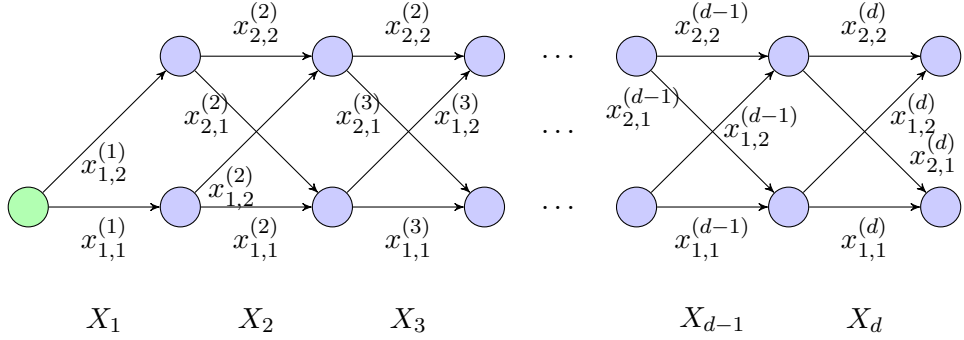
\begin{figure}
  \begin{center}
\begin{tikzpicture}[scale=1]
  \centering 
  \node[initialvertex] (v0) at (0,0) {};
  \node[vertex] (v11) at (2,0) {};
  \node[vertex] (v21) at (2,2) {};
  \node[vertex] (v12) at (4,0) {};
  \node[vertex] (v22) at (4,2) {};
  \node[vertex] (v13) at (6,0) {};
  \node[vertex] (v23) at (6,2) {};
  \node (mid1) at (7,0) {$\dots$};
  \node (mid2) at (7,2) {$\dots$};
  \node (mid3) at (7,1) {$\dots$};
  \node[vertex] (v15) at (8,0) {};
  \node[vertex] (v25) at (8,2) {};
  \node[vertex] (v16) at (10,0) {};
  \node[vertex] (v26) at (10,2) {};
  \node[vertex] (v17) at (12,0) {};
  \node[vertex] (v27) at (12,2) {};
  \path[->]
  (v0) edge node[below] {$x^{(1)}_{1,1}$} (v11)
  (v0) edge node[below] {$x^{(1)}_{1,2}$} (v21)
  (v11) edge node[below,pos=.3] {$x^{(2)}_{1,2}$} (v22)
  (v11) edge node[below] {$x^{(2)}_{1,1}$} (v12)
  (v21) edge node[above] {$x^{(2)}_{2,2}$} (v22)
  (v21) edge node[below, pos=.1] {$x^{(2)}_{2,1}$} (v12)
  (v12) edge node[below] {$x^{(3)}_{1,1}$} (v13)
  (v12) edge node[below, pos=.9] {$x^{(3)}_{1,2}$} (v23)
  (v22) edge node[below, pos=.1] {$x^{(3)}_{2,1}$} (v13)
  (v22) edge node[above] {$x^{(2)}_{2,2}$} (v23)
  (v15) edge node[below] {$x^{(d-1)}_{1,1}$} (v16)
  (v15) edge node[below, pos=.9, yshift=-5pt] {$x^{(d-1)}_{1,2}$} (v26)
  (v25) edge node[below, pos=.05, xshift=-5pt] {$x^{(d-1)}_{2,1}$} (v16)
  (v25) edge node[above] {$x^{(d-1)}_{2,2}$} (v26)
  (v16) edge node[below] {$x^{(d)}_{1,1}$} (v17)
  (v16) edge node[below, pos=.9] {$x^{(d)}_{1,2}$} (v27)
  (v26) edge node[above, pos=.95, xshift=5pt] {$x^{(d)}_{2,1}$}(v17)
  (v26) edge node[above] {$x^{(d)}_{2,2}$} (v27);

  \node (label1) at (1,-1.5) {$X_1$};
  \node (label2) at (3,-1.5) {$X_2$};
  \node (label3) at (5,-1.5) {$X_3$};
  \node (labeldminus1) at (9,-1.5) {$X_{d-1}$};
  \node (labeld) at (11,-1.5) {$X_d$};
\end{tikzpicture}
\caption{The directed acyclic graph $G_d$ that defines the polynomial $\IMM_d$ with its labeling.}
\label{fig:IMM}
\end{center}
\end{figure}

Another standard way of defining the polynomial $\IMM_d$ is via graphs. Define the edge-labelled directed acyclic graph $G_d= (V,E,\lambda)$ as follows: the vertex set $V$ is defined to be the disjoint union of vertex sets $V^{(0)},\ldots,V^{(d)}$ where $V^{(i)} = \{v^{(i)}_1,v^{(i)}_2\}.$ The edge set $E$ is the set of all possible edges from some set $V^{(i)}$ to $V^{(i+1)}$ (for $i < d$). The labelling function $\lambda:E\rightarrow X$ is defined by $\lambda((v^{(i)}_j,v^{(i+1)}_k)) = x^{(i+1)}_{j,k}.$ 
See Figure~\ref{fig:IMM} for a depiction of this graph.

Given a path $\pi$ in the graph $G_d$, $\lambda(\pi)$ is defined to be the product of all labels of edges in $\pi$. In this notation, $\IMM_d$ can be seen to be the following.
\begin{equation}
\label{eq:imm-def}
\IMM_d = \sum_{\substack{\text{paths $\pi$ from $v^{(0)}_1$}\\\text{to $v^{(d)}_1$ or $v^{(d)}_2$}}} \lambda(\pi) = \sum_{\pi_1,\ldots,\pi_d\in \{1,2\}} x^{(1)}_{1,\pi_1}x^{(2)}_{\pi_1,\pi_2}\cdots x^{(d)}_{\pi_{d-1},\pi_d}
\end{equation}

\subsection{Multilinear formulas and circuits}
\label{sec:mlsetup}

We refer the reader to the standard resources (e.g.~\cite{sy,github}) for basic definitions related to algebraic circuits and formulas. Having said that, we do make a few remarks.

\begin{itemize}
\item All the gates in our formulas and circuits will be allowed to have \emph{unbounded} fan-in.
\item  The size of a formula or circuit will refer to the number of gates (including input gates) in it, and depth of the formula or circuit will refer to the maximum number of product gates on a path from the input gate to output gate.
\item  Further, the \emph{product-depth} of the formula or circuit (as in \cite{ry08}) will refer to the maximum number of product gates on a path from the input gate to output gate. Note that the product depth of a formula or circuit can be assumed to be within a factor of two of the overall depth (by collapsing sum gates if necessary). 
\end{itemize}

\subparagraph*{Multilinear circuits and formulas.} An algebraic formula $F$ (resp. circuit $C$) computing a polynomial from $\F[X]$ is said to be \emph{multilinear} if each gate in the formula (resp. circuit) computes a multilinear polynomial. Moreover, a formula $F$ is said to be \emph{syntactic multilinear} if for each multiplication gate $\Phi$ of $F$ with children $\Psi_1,\ldots,\Psi_t$, we have 
$\supp(\Psi_i)\cap \supp(\Psi_j) = \emptyset \text{ for each $i\neq j$},$ 
where $\supp(\Phi)$ denotes the set of variables that appear in the subformula rooted at $\Phi.$ Finally, for $\Delta \geq 1$, we say that a multilinear formula (resp. circuit) is a $(\Sigma\Pi)^{\Delta}\Sigma$ formula (resp. circuit) if the output gate is a sum gate and along any path, the sum and product gates alternate, with each product gate appearing exactly $\Delta$ times and the bottom gate being a sum gate. We can define $(\Sigma\Pi)^{\Delta}, \Sigma\Pi\Sigma, \Sigma\Pi\Sigma\Pi$ formulas and circuits similarly. 

For a gate $\Phi$ in a syntactically multilinear formula, we define a set of variables $\Vars(\Phi)$ in a top-down fashion as follows.
\begin{definition}\label{def:Vars}
Let $C$ be a syntactically multilinear formula computing a polynomial on the variable set $X$.  For the output gate $\Phi$, which is a sum gate, we define $\Vars(\Phi) = X$. If $\Phi$ is a sum gate with children $\Psi_1, \ldots, \Psi_k$ and $\Vars(\Phi) =S \subseteq X$, then for each $1\leq i \leq k$, $\Vars(\Psi_i) = S$. If $\Phi$ is a product gate with children $\Psi_1, \ldots \Psi_k$ and $\Vars(\Phi) = S \subseteq X$, then $\Vars(\Psi_i) = \supp(\Psi_i)$ for $1 \leq i \leq k-1$ and $\Vars(\Psi_k) = S \setminus \left(\cup_{i=1}^{k-1} \Vars(\Psi_i)\right)$.
\end{definition}

It is easy to see that $\Vars(\cdot)$ satisfies the properties listed in the following proposition.

\begin{proposition}
\label{prop:vars}
For each gate $\Phi$ in a syntactically multilinear formula $C$, let $\Vars(\Phi)$ be defined as in Definition~\ref{def:Vars} above. 

\begin{enumerate}
\item For any gate $\Phi$ in $C$, $\supp(\Phi) \subseteq \Vars(\Phi)$. 
\item If $\Phi$ is an sum gate, with children $\Psi_1, \Psi_2, \ldots, \Psi_k$, then $\forall i \in [k]$, $\Vars(\Psi_i) = \Vars(\Phi)$. 
\item If $\Phi$ is a product gate, with children $\Psi_1, \Psi_2, \ldots, \Psi_k$, then $\Vars(\Phi) = \cup_{i=1}^k \Vars(\Psi_i)$ and the sets $\Vars(\Psi_i)$ ($i\in [k]$) are pairwise disjoint.
\end{enumerate}
\end{proposition} 

We will use the following structural results that convert general multilinear circuits (resp. formulas) to $(\Sigma\Pi)^\Delta \Sigma$ circuits (resp. formulas). 

\begin{lemma}[Raz and Yehudayoff~\cite{ry09}, Claims 2.3 and 2.4]
\label{lem:RY-nf-formulas}
For any multilinear formula $F$ of product depth at most $\Delta$ and size at most $s$, there is a syntactic multilinear $(\Sigma\Pi)^{\Delta}\Sigma$ formula $F'$ of size at most $(\Delta+1)^2\cdot s$ computing the same polynomial as $F$. 
\end{lemma}

\begin{lemma}[Raz and Yehudayoff~\cite{ry09}, Lemma 2.1]
\label{lem:RY-nf-ckts}
For any multilinear circuit $C$ of product depth at most $\Delta$ and size at most $s$, there is a syntactic multilinear $(\Sigma\Pi)^{\Delta}\Sigma$ formula $F$ of size at most $(\Delta+1)^2\cdot s^{2\Delta+1}$ computing the same polynomial as $C$. 
\end{lemma}

We will also need the following structural result.  

\begin{lemma}[Raz, Shpilka and Yehudayoff~\cite{RSY08}, Claim 5.6]
\label{lem:decompose}
Let $F$ be a syntactic multilinear formula computing a polynomial $f$ and let $\Phi$ be any gate in $F$ computing a polynomial $g$. Then $f$ can be written as $f = Ag + B$, where $A \in \F[X \setminus \Vars(\Phi)]$, $B \in \F[X]$ and $B$ is computed by replacing $\Phi$ with a $0$ in $F$. 
\end{lemma}

A standard divide-and-conquer approach yields the best-known multilinear formulas and circuits for $\IMM_d$ for all depths. 

\begin{lemma}
\label{lem:IMMubds}
For each $\Delta \leq \log d$,\footnote{All our logarithms will be to base $2$.} $\IMM_d$ is computed by a syntactic multilinear $(\Sigma\Pi)^{\Delta}$ circuit $C_\Delta$ of size at most $d^{O(1)}\cdot 2^{O(d^{1/\Delta})}$ and a syntactic multilinear $(\Sigma\Pi)^{\Delta}$ formula $F_{\Delta}$ of size at most $2^{O(\Delta d^{1/\Delta})}.$
\end{lemma}
\begin{proof}[Proof sketch.]
  We will first recursively construct $C_\Delta$. Let us recall that the $\IMM_d$ polynomial is defined over the matrices $M^{(1)}, M^{(2)}, \dots, M^{(d)}$. Let us divide these matrices into $t=d^{1/\Delta}$ contiguous blocks of size $d/t$ each, say $B_1, B_2, \dots, B_t$. The polynomial $\IMM_d$ can now be expressed in terms of those blocks of matrices as follows.
  \begin{equation}
    \label{eqn:imm-selfreduction}
    \IMM_d = \sum_{(u_1, u_2, \dots, u_{t})\in \{1,2\}^{t}}P^{(1)}_{1,u_1}P^{(2)}_{u_1,u_2}\dots P^{(t)}_{u_{t-1},u_t},
  \end{equation}
where $P^{(i)}_{u_{i-1},{u_i}}$ is the $(u_{i-1}, u_i)$-th entry of the product of the matrices in the $i$-th block. (In the special case $i=1$, take $u_0 =1$.)  It is important to note that each of the polynomials $P^{(i+1)}_{u_{i},u_{i+1}}$, defined over the block $B_{i+1}$, for all $i\in[t-1]$, is (almost) an instance of $\IMM_{d/t}$ over the  suitable set of variables. This enables us to recurse for $\Delta$ steps while obtaining a $\Sigma\Pi$ layer at each step. Thus, we get the following recursive formula for the size of the $(\Sigma\Pi)^\Delta$ circuit computing $\IMM_d$.
  \begin{align*}
    s(d, \Delta) \leq  t^{O(1)}\cdot(s(d/t, \Delta -1)) + 2^{O(t)}.
  \end{align*}
  Upon unfurling, this recursion gives us the needed bound of $d^{O(1)}\cdot 2^{O(d^{1/\Delta})}$.

  Let us now construct a multilinear formula for this polynomial\footnote{It is important to note that simple replication of nodes in $C_{\Delta}$ would prove to be wasteful.}. Consider the polynomial expression in Equation~\ref{eqn:imm-selfreduction}. If each of the polynomials $P^{(k)}_{u_{i}, u_i+1}$ is replaced by a variable, say $y^{(k)}_{u_i, u_{i+1}}$, the computation is of an instance of $\IMM_{t}$ over the variables $\{y^{(k)}_{u_i, u_{i+1}}\}$. Then there is a $\Sigma\Pi$ formula $F_1$ (say) that computes $\IMM_{t}$ of size $c^t$ (for some constant $c$) whose leaves are labelled by the variables of the form $y_{u_{i}, u_{i+1}}$. Since each of these leaves is an instance of $\IMM_{d/t}$ (over a suitable set of variables) themselves, this can further be partitioned into $t$ contiguous chunks of $d/t^2$ many matrices each. This when expressed as a $\Sigma\Pi$ formula (by introducing new variables) is of size $c^t$. By substituting the formulas obtained now for each of the polynomials $P^{(k)}_{u_{i}, u_i+1}$ into $F_1$ suitably to obtain a formula $F_2$ (say), of size $c^t \cdot c^t = c^{2t}$. This is a $\Sigma\Pi\Sigma\Pi$ formula whose leaves are variables corresponding to the instances of $\IMM_{d/t^2}$. Continuing this process for $\Delta$ steps gives us a $(\Sigma\Pi)^\Delta$ formula $F_{\Delta}$ with $2^{O(\Delta t)} = 2^{O(\Delta d^{1/\Delta})}$ many leaves.
\end{proof}

We will show that the above bounds are nearly tight in the multilinear setting. If we remove the multilinear restriction on $(\Sigma\Pi)^{\Delta}\Sigma$ formulas computing $\IMM_d$, we can get better upper bounds, as long as the underlying field has characteristic zero. 

\begin{lemma}[follows from~\cite{GKKSdepth3}]\label{lem:IMMgenubd}
Let $\F$ be a field of characteristic zero.  For each $\Delta \leq \log d$, $\IMM_d$ has a $(\Sigma\Pi)^{\Delta}\Sigma$ formula $F_\Delta$ of size at most $2^{O(\Delta d^{1/(2\Delta)})}$.
\end{lemma}
\begin{proof}[Proof sketch of Lemma~\ref{lem:IMMgenubd}] As in the proof of Lemma~\ref{lem:IMMubds}, we crucially use the self-reducibility of $\IMM_d$. We need the following claim (implicit in Gupta et al.~\cite{GKKSdepth3}) to prove this lemma.
\begin{claim}
  \label{clm:IMMt-GKKS}
  For $t>1$, $\IMM_t$ has a depth three non-multilinear formula of size at most $2^{O(\sqrt{t})}$ over any field of characteristic $0$.
\end{claim}

\begin{proof}[Proof of Claim~\ref{clm:IMMt-GKKS}]
  
Applying  Lemma~\ref{lem:IMMubds} with $\Delta = 2$ yields a $\Sigma\Pi\Sigma\Pi$ formula $F$ for $\IMM_d$  of size $2^{O(\sqrt{d})}.$ It can be checked from the proof of Lemma~\ref{lem:IMMubds} that this formula satisfies the additional property that all the product gates in the formula have fan-in $O(\sqrt{t}).$

   Over any field $\F$ of characteristic zero\footnote{It also works if the characteristic field $\F$ is positive but suitably large.}, Gupta et al.~\cite{GKKSdepth3} showed that any $\Sigma\Pi\Sigma\Pi$ formula of size $s$ where all product gates have fan-in at most $k$ can be converted into a $\Sigma\Pi\Sigma$ formula of size $\poly(s)\cdot 2^{O(k)}$. Applying this result to the formula $F$ obtained above, we get that $\IMM_t$ can indeed be computed by a $\Sigma\Pi\Sigma$ formula of size at at most $2^{O(\sqrt{t})}$, over any field $\F$ of characteristic zero. 
 \end{proof}
 
Consider the self reduction of the $\IMM_d$ polynomial as follows. Split the $d$ matrices being multiplied in $\IMM_d$ into $t=d^{1/\Delta}$ blocks with $d/t$ many matrices each. Let the variables $Y=\{y^{(k)}_{u,v}\ |\ k\in[t], u,v\in\{1,2\}\}$ correspond to the polynomials $\mc{P} = \{P^{(k)}_{u,v}\ |\ k\in[t], u,v\in\{1,2\}\}$ as defined in Lemma~\ref{lem:IMMubds}.

Let $\IMM_{t}(Y)$ be the polynomial that is obtained by replacing all the polynomials $P^{k}_{i,j}$ above with the corresponding variables. From Claim~\ref{clm:IMMt-GKKS}, we know that $\IMM_t(Y)$ has a $\Sigma\Pi\Sigma$ formula $F_1$ of size at most $c^{\sqrt{t}}$ for some constant $c$. It is easy to see that $\IMM_d$ can now be obtained by substituting for each of the variables in $Y$ (which appear at the leaves of $F_1$) with the corresponding polynomial in $\mc{P}$.  Using the above mentioned self-reducibility property, we shall self-reduce $\IMM_{d/t}$ again and obtain an instance of $\IMM_{t}$ over suitable set of new variables. This too has a $\Sigma\Pi\Sigma$ formula of size $c^{\sqrt{t}}$. The total number of leaves of the new $(\Sigma\Pi\Sigma)(\Sigma\Pi\Sigma)$ formula $F_2$ (say) is $c^{\sqrt{t}}\cdot c^{\sqrt{t}} = c^{2\sqrt{t}}$. Continuing this process for $\Delta$ steps yields us a $(\Sigma\Pi\Sigma)^\Delta$ formula of size $2^{O(\Delta \sqrt{t})} = 2^{O(\Delta d^{1/(2\Delta)})}$. We can merge two consecutive layers of $\Sigma$ gates into one layer of $\Sigma$ gates and thus obtain a $(\Sigma\Pi)^\Delta\Sigma$ formula $F_{\Delta}$ of size $2^{O(\Delta d^{1/(2\Delta)})}\,$. 
\end{proof}

\section{Lower bounds for multilinear formulas and circuits computing $\IMM_d$}
\label{sec:main}


The main theorem of this section is the following lower bound. 

\begin{theorem}
\label{thm:mainlbd}
Let $d\geq 1$ be a growing parameter and fix any $\Delta \leq \log d.$ Any syntactic multilinear $(\Sigma\Pi)^{\Delta}\Sigma$ formula for $\IMM_d$ must have a size of $2^{\Omega(\Delta d^{1/\Delta})}.$
\end{theorem}

Putting together Theorem~\ref{thm:mainlbd} with Lemmas~\ref{lem:RY-nf-formulas} and \ref{lem:RY-nf-ckts}, we have the following (immediate) corollaries. 

\begin{corollary}
\label{cor:cktlbd}
Let $d\geq 1$ be a growing parameter and fix any $\Delta \leq \log d/\log \log d.$ Any  multilinear circuit of product-depth $\Delta$ for $\IMM_d$ must have a size of $2^{\Omega(d^{1/\Delta})}.$ In particular, any polynomial-sized multilinear circuit for $\IMM_d$ must have product-depth $\Omega(\log d/\log \log d).$
\end{corollary}

\begin{corollary}
\label{cor:formlbd}
Let $d\geq 1$ be a growing parameter and fix any $\Delta \leq \log d.$ Any  multilinear $(\Sigma\Pi)^{\Delta}\Sigma$ formula for $\IMM_d$ must have size $2^{\Omega(\Delta d^{1/\Delta})}.$ In particular, any polynomial-sized multilinear formula for $\IMM_d$ must have product-depth $\Omega(\log d).$
\end{corollary}

Since the product-depth of a formula is at most its depth, Lemma~\ref{lem:IMMubds} and Corollary~\ref{cor:formlbd} further imply the following. 
\begin{corollary}[Tightness of Brent's depth-reduction for multilinear formulas]
\label{cor:Brent}
For each $d\geq 1$, there is an explicit polynomial $F_d$ defined on $O(d)$ variables such that $F_d$ has a multilinear formula of size $d^{O(1)}$, but any formula of depth $o(\log d)$ for $F_d$ must have a size of $d^{\omega(1)}.$
\end{corollary}

Choosing parameters carefully, we also obtain the following.

\begin{corollary}[Separation of multilinear formulas and general formulas over zero characteristic]\label{cor:sep}
Let $\F$ be a field of characteristic zero. Let $s\in \mathbb{N}$ be any growing parameter and $\Delta \in \mathbb{N}$ be such that $\Delta \leq o(\log s)$. There is an explicit multilinear polynomial $F_{s,\Delta}$ such that $F_{s,\Delta}$ has a $(\Sigma\Pi)^{\Delta}\Sigma$ formula of size $s$, but any $(\Sigma\Pi)^{\Delta}\Sigma$ multilinear formula for $F_{s,\Delta}$ must have a size of $s^{\omega(1)}.$
\end{corollary}
\begin{proof}
We choose the polynomial $F_{s,\Delta}$ to be $\IMM_d$ for suitable $d$ and then simply apply Theorem~\ref{thm:mainlbd} and Lemma~\ref{lem:IMMgenubd} to obtain the result. Details follow.

Say $\Delta = \log s/f(s)$ for some $f(s) = \omega(1).$ By Lemma~\ref{lem:IMMgenubd}, for any $d$, $\IMM_d$ has a product-depth $\Delta$ formula of size $s(d,\Delta) = 2^{O(\Delta d^{1/2\Delta})}$; we choose $d$ so that $s(d,\Delta)=s.$ It can be checked that for $d = \Theta(f(s))^{2\Delta}$, this is indeed the case. 

Having chosen $d$ as above, we define $F_{s,\Delta} = \IMM_d.$ Clearly, $F_{s,\Delta}$ has a (non-multilinear) formula of product-depth $\Delta$ and size at most $s$. On the other hand, by Theorem~\ref{thm:mainlbd}, any multilinear product-depth $\Delta$ formula for $\IMM_d$ must have size at least
\begin{align*}
2^{\Omega(\Delta d^{1/\Delta})} = s^{\Omega(d^{1/2\Delta})} = s^{\Omega(f(s))} = s^{\omega(1)},
\end{align*}
which proves the claim.

It can also be proved similarly that for $d$ as chosen above, $\IMM_d$ in fact has no multilinear formulas of size $s^{O(1)}$ and product-depth up to $(2-\varepsilon) \Delta$ for any absolute constant $\varepsilon.$ 
\end{proof}

\section{Proof of Theorem~\ref{thm:mainlbd}}
\label{sec:mainpf}

Our proof follows a two-step argument as in~\cite{Raz,ry09} (see the exposition in~\cite[Section 3.6]{sy}).

\subsection*{Step1 -- The product lemma} The first step is a ``product-lemma'' for multilinear formulas. 
%

Formally, define a polynomial $f\in \F[X]$ to be a \emph{$t$-product polynomial} if we can write $f$ as $f_1\cdots f_t$\,, where we can find a partition of $X$ into non-empty sets $X^f_1,\ldots,X^f_t$ such that $f_i$ is a multilinear polynomial from $\F[X^f_i].$\footnote{Note that we do not need $f_i$ to depend non-trivially on all (or any) of the variables in $X_i^f$. } We say that $X_i^f$ is the set \emph{ascribed} to $f_i$ in the $t$-product polynomial $f$.  We use $\vars(f_i)$ (with a slight abuse of notation)\footnote{$\Vars(\cdot)$ is used to describe variables ascribed to gates in a circuit as well as to denote variables ascribed to polynomials.} to denote $X^f_i$. We drop $f$ from the superscript if $f$ is clear from the context.

We define $f\in \F[X]$ to be \emph{$r$-\simple} if $f = L_1\cdots L_{r'}\cdot G$, where $r'\leq r$, is an $(r'+1)$-product polynomial where $L_1,\ldots,L_{r'}$ are polynomials of degree at most $1$, the sets $X_1^f,\ldots,X_{r'}^f$ ascribed to these linear polynomials satisfy $\left|\bigcup_{i\leq r'} X_i^f\right|\geq 400 r$. We prove the following.
\begin{lemma}
\label{lem:prod}
Let $\Delta \leq \log d.$ Assume that $f\in \F[X]$ can be computed by a syntactic multilinear $(\Sigma\Pi)^{\Delta}\Sigma$ formula $F$ of size at most $s$. Then, $f$ is the sum of at most $s$ many $t$-product polynomials and at most $s$ many $t$-\simple\ polynomials for $t = \Omega(\Delta d^{1/\Delta}).$
\end{lemma}

While our proof of the product lemma is motivated by earlier work~\cite{sy,hy,ry09}, we give slightly better parameters, which turns out to be crucial for proving tight lower bounds for formulas. In particular, \cite[Claim 5.5]{ry09} 
yields the above with $t = \Omega(d^{1/\Delta}).$ 

\begin{proof}[Proof of Lemma~\ref{lem:prod}]
Let $F$ be the $(\Sigma\Pi)^\Delta\Sigma$ syntactic multilinear formula of size at most $s$ computing $f$. We use layer $i$ to denote the layer at distance $i$ from the leaves. So in our formula, layer $1$ is a sum layer, layer $2$ is a product layer and so on.  Let $r = \Delta d^{1/\Delta}/400.$ 

We will prove by induction on the size $s$ of the formula $F$ that $f$ is the sum of at most $s$ polynomials, each of which is either a $t$-product polynomial or a $t$-simple polynomial for $t = \Delta d^{1/\Delta}/1000.$ 

The base case of the induction, corresponding to $s=0$, is trivial. 

\subparagraph*{Case 1:} Suppose there exists a gate $\Phi$ in layer $2$ such that $\Phi$ computes a polynomial $g$ and has fan-in at least $t$. Then we use Lemma~\ref{lem:decompose} and decompose $f$ as  $Ag +B$. Here $Ag$ is a $t$-product polynomial. Since $B$ is computed by a formula of size at most $s-1$, we are done by induction. 

\subparagraph*{Case 2:} Suppose the above case does not hold, i.e. all the gates at layer $2$ have a fan-in of at most $t$. Now, if there exists a gate $\Phi$ in layer $2$ such that $|\Vars(\Phi)| \geq 400 r$ then we will decompose $F$ using Lemma~\ref{lem:decompose} and obtain $f = Ag+H$, where $Ag$ is $t$-\simple\ since $|\Vars(\Phi)| \geq 400 r \geq  400t$. Again, since $H$ has a formula of size at most $s-1$, and we are done by induction.

\subparagraph*{Case 3:} Now assume that neither of the above cases is applicable. Since neither Case 1 nor Case 2 above is applicable to $F$, each gate $\Phi$ in layer $2$ satisfies $|\Vars(\Phi)| < p := 400r.$ This immediately implies that $\Delta \geq 2,$ since in the case of a $\Sigma\Pi\Sigma$ formula, we have $|\Vars(\Phi)| = n$ by Proposition~\ref{prop:vars} item 2 but $p = 400 r \leq d < n.$  

If $\Delta\geq 2,$ we use the following lemma. 

\begin{lemma}\label{lem:prod-depth}
Let $n,p\in \mathbb{N}$. Assume $2\leq \Delta \leq 2 \log(n/p).$ Let $f$ be computed by a  syntactically multilinear $\left(\Sigma\Pi\right)^{\Delta}\Sigma$ formula $F$  of size at most $s$ over a set of $n$ variables. Let $\Phi_1, \Phi_2, \ldots, \Phi_{s'}$, where $s'\leq s$, be the product gates at layer $2$ such that for all $i$, $|\Vars(\Phi_i)| \leq p$, then $f$ is the sum of at most $s$ many $T$-product polynomials where $T = (\Delta \left(n/p\right)^{1/(\Delta-1)})/100$. 
\end{lemma}

The above lemma is applicable in our situation since we have $\Delta \leq \log d,$ $n \geq 2d$, and hence $(n/p) = (n/400r) = n/(\Delta d^{1/\Delta}) \geq n/(2\sqrt{d}) \geq \sqrt{d}.$ Lemma~\ref{lem:prod-depth} now yields a decomposition of $f$ as a sum of at most $s$ many $T$-product polynomials where 
\begin{align*}
T = \Delta\cdot \frac{(n/400 r)^{\frac{1}{(\Delta-1)}}}{100}  \geq \frac{\Delta}{100} \cdot {\left(\frac{d}{\Delta d^{1/\Delta}}\right)^{\frac{1}{(\Delta-1)}}} =  \frac{\Delta d^{1/\Delta}}{100\Delta^{1/{\Delta-1}}} \geq \frac{\Delta d^{1/\Delta}}{200}\,.
\end{align*}
Since $T\geq t$, these $T$-product polynomials are also $t$-product polynomials. This finishes the proof of the claim modulo the proof of Lemma~\ref{lem:prod-depth}, which we present below.
\end{proof}

\begin{proof}[Proof of Lemma~\ref{lem:prod-depth}]
We shall prove by induction on the depth $\Delta$ that we can take $T = t(n,\Delta) =  (\Delta-1)\left(\left(n/p\right)^{1/(\Delta-1)}-1\right)$. Since $\Delta \leq 2\log(n/p),$ this implies that $T\geq \Delta (n/p)^{1/(\Delta-1)}/100.$ 

Let $X$ denote the set of all $n$ underlying variables. 

The base case is when $\Delta = 2$. Here, we have a $\Sigma\Pi\Sigma\Pi\Sigma$ formula such that for all $\Phi$ at layer $2$, $|\Vars(\Phi)| \leq p$. Let $\Psi$ be the output (sum) gate of the formula and $\Psi_1,\ldots,\Psi_r$ be the product gates feeding into it; further let $f_i$ be the polynomial computed by $\Psi_i$. We claim that each $f_i$ is an $(n/p)$-product polynomial. If this is true, we are done since $f = f_1 + \cdots + f_r$ and $r$ is at most $s$. 

To show that $f_i$ is an $(n/p)$-product polynomial, it suffices to show that each $\Psi_i$ has fan-in at least $(n/p).$ This follows since each $\Phi$ at layer $2$ satisfies $|\Vars(\Phi)|\leq p$ and for each sum gate $\Phi'$ at layer $3$, we have $\Vars(\Phi')=\Vars(\Phi)$ for any gate $\Phi$ at layer $2$ feeding into $\Phi'$ (Proposition~\ref{prop:vars} item 2). By Proposition~\ref{prop:vars} item 3, the fan-in of each $\Psi_i$ at layer $4$ must thus be at least $(n/p).$ This concludes the base case. 

Now consider $\Delta \geq 3$. Say we have a polynomial $f$ that is computed by a $(\Sigma\Pi)^{\Delta}\Sigma$ formula $F$ of size at most $s$ and top fan-in (say) $r$. Let $\Psi$ be the output gate of $F$ and $\Psi_1,\ldots,\Psi_r$ the product gates feeding into it; let $f_i$ be the polynomial computed by $\Psi_i.$ It suffices to show that each $f_i$ is the sum of at most $s_i$ many $t(n,\Delta)$-product polynomials, where $s_i$ is the size of the subformula rooted at $\Psi_i.$ We show this now.

Fix any $i\in [r]$. Let the children of $\Psi_i$ be $\Psi_{i,1},\ldots,\Psi_{i,k}.$ Since $X = \Vars(\Psi) = \bigcup_{j=1}^k \Vars(\Psi_{i,j})$ (Proposition~\ref{prop:vars} item 3), there must be some gate $\Psi_{i,j}$ feeding into $\Psi_i$ such that $|\Vars(\Psi_{i,j})|\geq n/k$; w.l.o.g., assume that $j=1$. Applying the induction hypothesis for depth $\Delta-1$ formulas to the polynomial $f_{i,1}\in \F[\Vars(\Psi_{i,1})]$ computed by the subformula rooted at $\Psi_{i,1}$, we obtain
\[
f_{i,1} = \sum_{\ell = 1}^{s_i} h_{i,1,\ell}
\]
where each $h_{i,1,\ell}$ is a $t(n/k,\Delta-1)$-product polynomial. Hence, we see that 
\[
f_i = f_{i,1}\cdots f_{i,k} = \sum_{\ell=1}^{s_i} h_{i,1,\ell} f_{i,2}\cdots f_{i,k}.
\]
Each term in the above decomposition of $f_i$ is a $t'$-product polynomial for $t' = t(n/k,\Delta-1)+(k-1)$ where $k$ is the fan-in of $f_i.$ Some calculus shows that the expression $t(n/k,\Delta-1)+(k-1)$ is minimized when $k = (n/p)^{1/\Delta-1}$. Plugging this into the expression gives $t' \geq t(n,\Delta).$ 

We have thus shown that no matter what $k$ is, $t' \geq t(n,\Delta),$ from which the induction step follows. 
\end{proof}
 

\subsection*{Step 2 -- Rank measure and the hard polynomial} The second step is to show that any such decomposition for $\IMM_d$ must have many terms. Our proof of this step is inspired by the proof of the multilinear formula lower bound of Raz~\cite{Raz} for the determinant and also the slightly weaker lower bound of Nisan and Wigderson~\cite{nw1997} for $\IMM_d$ in the \emph{set-multilinear} case. Following~\cite{Raz}, we define a suitable \emph{random restriction} of the $\IMM_d$ polynomial by assigning variables from the underlying variable set $X$ to $Y\cup Z\cup \{0,1\}$, where $Y$ and $Z$ are disjoint sets of new variables of equal size. The restriction sets distinct variables in $X$ to distinct variables in $Y\cup Z$ or constants, and hence preserves multilinearity. 

Having performed the restriction, we consider the \emph{partial derivative matrix} of the restricted polynomial, which is defined as follows. Let $g\in \F[Y\cup Z]$ be a multilinear polynomial. Define the $2^{|Y|}\times 2^{|Z|}$ matrix $\M_{(Y,Z)}(g)$ such that rows and columns are labelled by distinct multilinear monomials in $Y$ and $Z$ respectively and the $(m_1,m_2)$th entry of $\M_{(Y,Z)}(g)$ is the coefficient of the monomial $m_1\cdot m_2$ in $g$. 

Our restriction is defined to have the following two properties. 
\begin{enumerate}
\item The rank of $\M_{(Y,Z)}(g)$ is equal to its maximum possible value (i.e. $\min\{2^{|Y|},2^{|Z|}\}$) with probability $1$ where $g$ is the restricted version of $\IMM_d$.
\item On the other hand, let $f$ be either a $t$-product polynomial or a $t$-\simple\ polynomial, and let $f'$ denote its restriction under $\rho$. Then, the rank of $\M_{(Y,Z)}(f')$ is small with high probability.
\end{enumerate}

Now, if $\IMM_d$ has a $(\Sigma\Pi)^{\Delta}\Sigma$ formula $F$ of small size, then it is a sum of a small number of $t$-product and $t$-simple polynomials by Lemma~\ref{lem:prod} and hence by a union bound, we will be able to find a restriction under which the partial derivative matrices of each of the these polynomials has small rank. By the subadditivity of rank, this will imply that $\M_{(Y,Z)}(g)$ will itself have low rank, contradicting the first property of our restriction.

To make the above precise, we first define our restrictions. Let $\tilde{Y} = \{y_1,\ldots,y_d\}$ and $\tilde{Z} = \{z_1,\ldots,z_d\}$ be two disjoint sets of variables. A restriction $\rho$ is a function mapping variables $X$ to elements of $\tilde{Y}\cup \tilde{Z}\cup \{0,1\}.$ We consider the following process for sampling a random restriction.

\paragraph*{Notation.} Recall that $M^{(i)}$ is the $2\times 2$ matrix whose $(u,v)$th entry is $x^{(i)}_{u,v}$. Let $I$ and $E$ denote the standard $2\times 2$ identity matrix and the $2\times 2$ flip permutation matrix respectively. 
For $a\in \{1,2\},$ we use $\overline{a}$ to denote the other element of the set. 

\begin{algorithm}
    \caption{Sampling algorithm $\mc{S}$}
  \begin{algorithmic}[1]
    \STATE Choose $\pi$ uniformly at random from $\{1,2\}^d.$ Define $\pi(0) = 1.$
    \STATE Choose $a$ uniformly at random from $\{0,1\}^d.$ Let $A = \{i\ |\ a_i = 1\}.$
    \FOR{$i\in [d]$}
    \STATE Let $b_i = 0$ if $\pi(i-1) = \pi(i)$ and $1$ if $\pi(i-1) \neq \pi(i)$. 
    \ENDFOR
    \FOR{$i =1$ to $d$}
    \IF{$i\not\in A$}
    \STATE Choose $\rho|_{X^{(i)}}$ such that $M^{(i)}$ is $I$ if $b_i=0$ and $E$ if $b_i=1.$ (In particular, all variables are set to constants from $\{0,1\}$.)
    \ELSIF{$i\in A$ and $i$ is the $j$th smallest element of $A$ for odd $j$}
    \STATE Fix
    \[\rho(x^{(i)}_{u,v}) = \left\{
        \begin{array}{ll}
          y_{\lceil j/2\rceil} & \text{if $u = \pi(i-1)$ and $v= \pi(i)$,}\\
          1 & \text{if $u = \pi(i-1)$ and $v= \overline{\pi(i)}$,}\\
          0 & \text{otherwise.}
        \end{array}\right.
    \]
    \ELSE
    \STATE Now, $i\in A$ and $i$ is the $j$th smallest element of $A$ for even $j$. We fix 
    \[\rho(x^{(i)}_{u,v}) = \left\{
        \begin{array}{ll}
          z_{{j}/{2}} & \text{if $u = \pi(i-1)$ and $v= \pi(i)$,}\\
          1 & \text{if $u = \overline{\pi(i-1)}$ and $v= \pi(i)$,}\\
          0 & \text{otherwise.}
        \end{array}\right.
    \]
    \ENDIF
    \ENDFOR  
  \end{algorithmic}
\label{alg:sampling}
\end{algorithm}

  

We give a procedure $\mc{S}$ for sampling a random restriction $\rho:X \rightarrow \tilde{Y}\cup \tilde{Z}\cup \{0,1\}$ in Algorithm~\ref{alg:sampling}. 
Based on the output $\rho$ of $\mc{S}$, we define the (random) sets $Y = \tilde{Y}\cap \Img(\rho)$ and $Z = \tilde{Z}\cap \Img(\rho)$. Let $m = m(\rho) = \min\{|Y|,|Z|\}$.

We observe the following simple properties of $\rho$.
\begin{observation}
\label{obs:rho}
The restriction $\rho$ satisfies the following.
\begin{enumerate}
\item $|Y| = \lceil |A|/2\rceil$ and $|Z| = \lfloor |A|/2\rfloor$. Hence,  $|Z|\leq |Y|\leq |Z|+1$ and $m = |Z|.$
\item Distinct variables in $X$ cannot be mapped to the same variable in $Y\cup Z.$
\item Only the variables of the form $x^{(i)}_{\pi(i-1),\pi(i)}$ can be set to variables in $Y\cup Z$ by $\rho$. The rest are set to constants. 
\end{enumerate}
\end{observation}

Note that $b$ is distributed uniformly over $\{0,1\}^d.$
Given a polynomial $f\in \F[X]$, the restriction $\rho$ yields a natural polynomial $f|_\rho\in \F[Y\cup Z]$ by substitution. Note, moreover, that if $f$ is multilinear then so is $f|_\rho$ since distinct variables in $X$ cannot be mapped to the same variable in $Y\cup Z$ (Observation~\ref{obs:rho}).

\begin{lemma}
\label{lem:rand-rest}
Let us assume that $\rho$ is sampled as above. Then we have the following:
\begin{enumerate}
\item $\rank(\M_{(Y,Z)}(\IMM_d|_\rho)) = 2^m$ with probability $1$. 
\item If $f\in \F[X]$ is any $t$-product polynomial, then for some absolute constant $\varepsilon > 0,$ $$\prob{}{\rank(\M_{(Y,Z)}(f|_\rho)) \geq 2^{m - \varepsilon t}} \leq \frac{1}{2^{\Omega(t)}}.$$
\item If $f\in \F[X]$ is any $r$-\simple\ polynomial, then for some absolute constant $\delta > 0,$
$$
\prob{}{\rank(\M_{(Y,Z)}(f|_\rho)) \geq 2^{m - \delta r}} \leq \frac{1}{2^{\Omega(r)}}.
$$
\end{enumerate}
\end{lemma}

\noindent
Given Lemmas~\ref{lem:prod} and \ref{lem:rand-rest}, we can finish the proof of Theorem~\ref{thm:mainlbd} as follows.

\begin{proof}[Proof of Theorem~\ref{thm:mainlbd} assuming Lemma~\ref{lem:rand-rest}]
  Assume that $\IMM_d$ is computed by a syntactic mulitlinear $(\Sigma\Pi)^{\Delta}\Sigma$ formula $F$ of size at most $s$. By Lemma~\ref{lem:prod}, we get that $f$ can be expressed as a sum of at most $2s$ many summands, say $f_1,f_2, \dots, f_s$ and $g_1,g_2,\ldots,g_s$, where each summand $f_i$ is a $t$-product polynomial and each summand $g_j$ is a $t$-\simple\  polynomial for $t = \Omega(\Delta d^{1/\Delta})$.
  
  For each $i\in [s],$ Lemma~\ref{lem:rand-rest} implies that 
  \begin{align*}
\prob{}{\rank\left(M_{(Y,Z)}(f_i|_{\rho})\right) \geq 2^{m - \varepsilon t} }\leq \frac{1}{2^{\Omega(t)}}\ \text{and}\  \prob{}{\rank\left(M_{(Y,Z)}(g_i|_{\rho})\right) \geq 2^{m - \delta t} }\leq \frac{1}{2^{\Omega(t)}},
  \end{align*} 
  where $\varepsilon$ and $\delta$ are absolute constants. 
  
  Thus, unless $s \geq 2^{\Omega(t)},$ we see by a union bound that there exists a $\rho$ such that for each $i\in [s]$, $\rank\left(M_{(Y,Z)}(f_i|_{\rho})\right)\leq 2^{m-\varepsilon t}$ and $\rank\left(M_{(Y,Z)}(g_i|_{\rho})\right)  \leq 2^{m-\delta t}.$ For such a $\rho$, we have
 \[
 \rank(M_{(Y,Z)}(F|_{\rho})) \leq 2^m \cdot \left(\frac{s}{2^{\varepsilon t}} + \frac{s}{2^{\delta t}}\right) < 2^m
 \]
  unless $s \geq 2^{\Omega(t)}.$
  
  From Lemma~\ref{lem:rand-rest}, we also know that for \emph{any} choice of $\rho$ in the sampling algorithm $\mc{S},$ we have $\rank(M_{(Y,Z)}(\IMM_d|_\rho)) \geq 2^m.$ In particular, since $F$ computes $\IMM_d$, we must have $s \geq 2^{\Omega(t)} = 2^{\Omega(\Delta d^{1/\Delta})}.$
\end{proof}

\subsection{Proof of Lemma~\ref{lem:rand-rest}}
\noindent
\subsubsection*{Part 1: $\IMM_d$ has high rank} 
\label{sec:p1}
Let $\pi \in \{1,2\}^d$ and $a \in \{0,1\}^d$ be arbitrary. Note that in our sampling algorithm, $\rho, A, b$ are completely determined given $\pi$ and $a$. 

Let us now examine the effect of $\rho$ on $\IMM_d$. We take the graph theoretic view of the polynomial $\IMM_d$ as given in Section~\ref{sec:basic}.

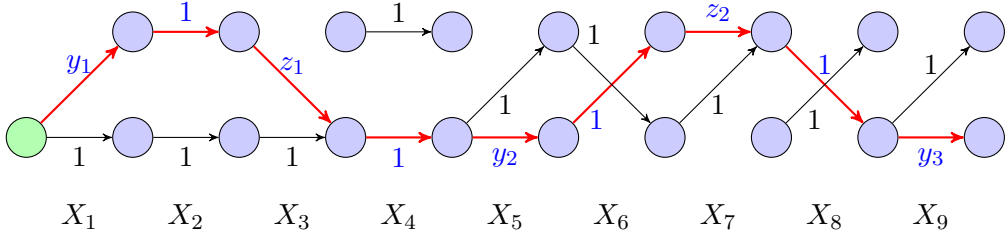
\begin{figure}
  \centering
  \begin{tikzpicture}[scale=0.7]
    \node[initialvertex] (v0) at (0,0) {};
    \foreach \x [evaluate=\x as \lbl using 2*\x-1] in {1,2,...,9}
    {
      \node[vertex] (v1\x) at (2*\x,0) {};
      \node[vertex] (v2\x) at (2*\x,2) {};
      \node (label\x) at (\lbl, -1.5) {$X_{\x}$};
    }

    \path[->]
    (v0) edge node[below] {$1$} (v11)
    (v11) edge node[below] {$1$} (v12)
    (v12) edge node[below] {$1$} (v13);

    \path[->, thick, red]
    (v0) edge node[en, above] {$y_1$} (v21)
    (v21) edge node[en, above] {$1$} (v22)
    (v22) edge node[en, above] {$z_1$} (v13)
    (v13) edge node[en, below] {$1$} (v14)
    (v14) edge node[en, below] {$y_2$} (v15)
    (v15) edge node[en, below, pos=0.3] {$1$} (v26)
    (v26) edge node[en, above] {$z_2$} (v27)
    (v27) edge node[en, above] {$1$} (v18)
    (v18) edge node[en, below] {$y_3$} (v19);
    
    \path[->] (v23) edge node[above] {$1$} (v24);

    \path[->]
    (v14) edge node[below] {$1$} (v25)
    (v25) edge node[above, pos=0.1, xshift=5pt] {$1$} (v16)
    (v17) edge node[below, pos=0.3, xshift=2pt] {$1$} (v28)
    (v16) edge node[below] {$1$} (v27);

    \path[->]
    (v18) edge node[above] {$1$} (v29)
    ;

  \end{tikzpicture}
  \caption{Effect of $\rho$ on $\IMM_9$ when the sampling algortithm $\mathcal{S}$ yields $\pi = (2,2,1,1,1,2,2,1,1)$ and $a = (1,0,1,0,1,0,1,0,1)$. Thus, $\IMM_9|_{\rho}$ in this case yields us $(1+y_1z_1)(1+y_2z_2)(1+y_3)$.}
  \label{fig:imm-rho}
\end{figure}

Figure~\ref{fig:imm-rho} illustrates how this restriction affects the variables labelling the edges of the graph $G_d$ defined in Section~\ref{sec:basic}. By substituting according to $\rho$ in (\ref{eq:imm-def}), we get that

 \[
 \begin{array}{ll}
 \IMM_d(X)|_\rho = & \left\{
 \begin{array}{lll}
 & \prod_{i=1}^m \left(1+y_iz_i\right) & \text{if } |A| = 2m \\
 & & \\
 &\prod_{i=1}^{m} \left(1+y_iz_i\right) \cdot \left(1+ y_{m+1} \right)& \text{if } |A| = 2m +1\,, \\
 \end{array} \right.
 \end{array}\]
 where $m = |Z|$. 
For any $S \subseteq [m]$, let $Z_S$ (resp., $Y_S$) denote the monomial $\prod_{i \in S} z_i$ (resp., $\prod_{i \in S} y_i$). Now consider the matrix $\M_{(Y,Z)}(\IMM_d|_\rho)$\,. We will simply use $\mathcal{M}$ to denote this matrix. For the sake of simplicity let us assume that $|A| =2m$. (The case when $|A|=2m+1$ is similar.) Let the rows and columns of $\mathcal{M}$ be labelled by the subsets of $[m]$ and let $\mathcal{M}(S,T)$ be the coefficient of $Y_S\cdot Z_T$ in  $\IMM_d|_\rho$. It is easy to see that $\mathcal{M}(S,T) = 0$ if $S \neq T$ and $1$ otherwise. That is, $\mathcal{M}$ is the Identity matrix of size $2^m\times 2^m$ and hence it has full rank.\footnote{If $|A|=2m+1$ then $\mathcal{M}$ has a $2^m\times 2^m$ sized Identity matrix as a submatrix.} \qed.

\subsubsection*{Part 2: $t$-product polynomials have low rank}  
\label{sec:p2}
We now prove that for a $t$-product polynomial $f$, $\rank(M_{(Y,Z)}(f|_\rho))$ is small with high probability. 

Let $f$ be a $t$-product polynomial, i.e. $f = f_1f_2\ldots f_t$. Let $\chi: X \rightarrow [t]$ be a coloring function, which assigns colors to all the variables in $X$, so that $\chi^{-1}(i) = X^f_i$, where $X_i^f$ is the variable set ascribed to $f_i$. That is, all the variables ascribed to $f_i$ are assigned color $i$ under the coloring function. To prove the lemma, we will first show that, with high probability (over the choice of $\pi$), a constant fraction of the $t$ colors appear along the path defined by $\pi$, i.e. along $(\pi(0),\pi(1)), (\pi(1),\pi(2)), \ldots, (\pi({d-1}),\pi(d))$. Given such a \emph{multi-colored path}, we will then show that with a high probability, over the choice of $a$, many of the colors have \emph{an imbalance}. A color is said to have an imbalance under $\rho$ if more variables from $X$ of that color are mapped to the $Y$ variables than the $Z$ variables or vice versa. We will then appeal to  arguments that are similar to those in~\cite{Raz,ry09,DMPY12} to conclude that the imbalance results in a low rank. 

\vspace*{5pt}
\noindent
\textbf{Variable coloring, $t$-product polynomials and imbalance.}
We start with some notation. Given a string $\pi \in \{1,2\}^d$, let the path defined by $\pi$ be the following sequence of pairs $(\pi(0),\pi(1)),$ $(\pi(1),\pi(2)), \ldots,$ $(\pi({d-1}),\pi(d))$ (we call it a path since these pairs correspond naturally to the edges of a path in the graph $G_d$ defined in Section~\ref{sec:basic}). We say that a color $\gamma \in [t]$ appears in layer $\ell\in [d]$ if there exists $u,v \in \{1,2\}$ such that $\gamma = \chi(x_{u,v}^{(\ell)})$. 

Let $C^{0} = \emptyset$ and let $C^i = C^{i-1} \cup \{\chi(x_{u,v}^{(i)}) \mid {u,v \in\{1,2\}}\}$ for $i \in [d]$, i.e., $C^i$ contains all the distinct colors appearing in layers $\{1,2 \ldots, i\}$. Therefore, $|C^d| = t$. 
We will also define $O^{2i+1}$ to be all the colors appearing in odd numbered layers up to $2i+1$, i.e. $O^{2i+1} = O^{2i-1} \cup \{\chi(x_{u,v}^{(2i+1)}) \mid {u,v \in\{1,2\}}\}$. Similarly, we define $E^{2i} = E^{2i-2} \cup \{\chi(x_{u,v}^{(2i)}) \mid {u,v \in\{1,2\}}\}$. 

Let $C_\pi^0 = \emptyset$ and $C_\pi^i = C_\pi^{i-1} \cup \{\chi(x_{(\pi(i-1),\pi(i))}^{(i)})\}$, i.e. $C_\pi^i$ contains all the distinct colors appearing along the path defined by $\pi$ up to layer $i$. 
We first observe a property of $C_\pi^d$ stated in the claim below.

\begin{claim}
\label{cl:many-color-paths}
If $|C^d| =t$, then $\prob{\pi}{|C_\pi^d| \leq t/100} \leq 1/2^{\Omega(t)}$\,.
\end{claim}

We will assume the claim and finish the proof of Part $2$ of Lemma~\ref{lem:rand-rest}. We will then prove the claim. 
The above claim shows that a lot of colors appear on the uniformly random path $\pi$ with high probability. Using this, we will now show that a constant fraction of these colors also exhibit an imbalance with a high probability. Using the multiplicativity of the rank, we will then show that the  imbalance for a large number of factors results in the low rank of the matrix $M_{Y,Z}(f|_\rho)$.

We will say that $\pi$ is good if $|C_\pi^d| > t/100$. Let $L = t/100$. The above claim shows that a random $\pi$ is good with high probability. In what follows, we condition on picking a good $\pi$. Let $a \in \{0,1\}^d$ be chosen uniformly at random as in the sampling algorithm. Let $\rho$ be defined as in the sampling algorithm for $\pi,a$. 

Let $\gamma \in C_{\pi}^d$ be a color that appears along $\pi$. Let $\pi_\gamma$ be the elements along the path defined by $\pi$ with color $\gamma$, i.e. $\pi_\gamma = \{(\pi(i-1), \pi(i)) \mid \chi(x_{(\pi(i-1), \pi(i))}^{(i)}) = \gamma\}$. Let $\rho(\pi_\gamma) =$ $\{\rho(x_{(\pi(i-1), \pi(i))}^{(i)}) \mid (\pi(i-1), \pi(i)) \in \pi_\gamma\}\cap (Y\cup Z)$\,. 
A color $\gamma \in [t]$ is said to have an imbalance w.r.t. $\rho$ if $||\rho(\pi_\gamma)\cap Y|- |\rho(\pi_\gamma)\cap Z|| \geq 1$. 

It is easy to see that if $|\rho(\pi_\gamma)|$ is odd, then $\gamma$ must have an imbalance w.r.t. $\rho$. Note that the former event is equivalent to the event that $\bigoplus_{i\in P_\gamma} a_i$ equals $1$ where $P_\gamma = \{i\ |\ (\pi(i-1),\pi(i)) \in \pi_\gamma\}.$
Hence for any $\gamma \in C_\pi^d$,
$\prob{}{\text{$\gamma$ has an imbalance with respect to $\rho$ along $\pi$}}  = 1/2$. 
%
Further, since $|C^d_\pi| \geq L$ and the events corresponding to distinct $\gamma\in C^d_{\pi}$ are mutually independent, the Chernoff bound implies
$
\prob{}{\text{at most $L/4$ colors have an imbalance with respect to $\rho$ along $\pi$}}  \leq 1/2^{\Omega{(L)}}.
$ \qed.

Assuming Claim~\ref{cl:many-color-paths} we are now done. We now present the proof of Claim~\ref{cl:many-color-paths}.

\begin{proof}[Proof of Claim~\ref{cl:many-color-paths}]
We define $O^{2i+1}_\pi$ to be all the colors appearing in odd numbered layers along $\pi$ up to the layer $2i+1$, i.e. $O^{2i+1}_\pi = O^{2i-1}_\pi \cup \{\chi(x_{(\pi(2i),\pi(2i+1))}^{(2i+1)})\}$. Similarly, we define $E^{2i}_\pi = E^{2i-2}_\pi \cup \{\chi(x_{(\pi(2i-1),\pi(2i))}^{(2i)})\}$. 

We know that $|C^d| = t$. Therefore, either $|O^d| \geq t/2$ or $|E^d| \geq t/2$. Let us assume without loss of generality that $|O^d| \geq t/2$. 
For this part of the proof, for the sake of simplicity, we will assume that $d$ is odd. The assumption can be easily removed by losing at most constant factors in the bound. 

Let $j_1, j_2, \ldots, j_\tau$ be odd indices such that for each $1 \leq i \leq \tau-1$\,, $|O^{j_i}| < |O^{j_{i+1}}|$\,, i.e. each $O^{j_i}$ has at least one new color. Let $\gamma_1, \gamma_2, \ldots, \gamma_\tau$ be colors which appear new in these sets. (If multiple new colors appear in a set then choose any one.)

 Let $W_i$ be the indicator random variable, which takes value $1$ if $|O^{j_i}_\pi| < |O^{j_{i+1}}_\pi|$ and $0$ otherwise, where $1 \leq i \leq \tau-1$. Then $\avg{}{W_i} = 1/4$ as the probability of the color $\gamma_i$ appearing in $O^{j_i}_\pi$ is equal to $1/4$. Note that the $W_i$s are independently distributed since they depend on distinct co-ordinates of $\pi$. 
Now $\avg{}{\sum_i W_i} \geq t/8$, as $|O^d| \geq t/2$. 
Now we get, 
\[\begin{array}{ll}
\prob{\pi}{|C_\pi^d| \leq t/100} & \leq \prob{\pi}{|O^{d}_\pi| \leq t/100} \\ 
& \leq \prob{\pi}{\sum_i W_i \leq t/100} \\
& \leq 1/2^{\Omega(t)}\,.\\
\end{array}
\]
As $|O^d_\pi| \leq |C^d_\pi|$, the first inequality follows. If the number of times a new color appears along $\pi$ within the odd layers is at most $t/100$, then $\sum_iW_i$ is also at most $t/100$, therefore we get the second inequality. Finally the last inequality follows by the Chernoff bound.

\end{proof}

\vspace*{5pt}
\noindent
\textbf{Imbalance implies low rank.}
\noindent
Let us recall that $f=f_1f_2\dots f_t$ is a $t$-product polynomial that is defined over the disjoint variable partition $X = X_1\cup X_2 \cup \dots \cup X_t$ such that $\abs{X_i} \geq 1$ for all $i\in[t]$. The following lemma (see, e.g.,~\cite{ry09}) will be useful in bounding $\rank(M_{(Y,Z)}(f|_\rho)).$
\begin{lemma}[\cite{ry09}, Proposition 2.5]
\label{lem:mult}
Let $g = g_1g_2\cdots g_t$ be a $t$-product polynomial over the set of variables $Y\cup Z$ where $\Vars(g_i) = Y_i \cup Z_i.$ Then $\rank(M_{(Y,Z)}(g)) = \prod_{i\in [t]}\rank(M_{(Y_i,Z_i)}(g_i)).$
\end{lemma}

From Lemma~\ref{lem:mult}, we get that $\rank(M_{(Y,Z)}(f|_\rho)) = \prod_{i=1}^{t}\rank(M_{(Y_i,Z_i)}(f_i|_\rho))$ where $Y_i = Y \cap \{\rho(x) | x\in X_i\}$ and $Z_i = Z \cap \{\rho(x) | x \in X_i\}\,.$ For all $i\in[t]$, from the definition it is clear that the rank of the matrix $M_{(Y_i,Z_i)}(f_i|_\rho)$ is upper bounded by $2^{\min\{|Y_i|,|Z_i|\}}\leq 2^{(\abs{Y_i}+\abs{Z_i})/2}$. Let us note that these disjoint partitions in the $t$-product polynomial correspond to the colors in the coloring $\chi$ with all variables in $X_i$ colored $i$.  Hence if color $i$ has imbalance w.r.t. $\rho$, then $\rank(M_{(Y_i,Z_i)}(f_i|_\rho))\leq 2^{\min\{|Y_i|,|Z_i|\}}\leq 2^{(\abs{Y_i}+\abs{Z_i}-1)/2}$. Thus, $\rank(M_{(Y,Z)}(f|_\rho)) \leq \prod_{i=1}^t 2^{(\abs{Y_i}+\abs{Z_i}-1)/2} = 2^{((\abs{Y}+\abs{Z})/2)-(\ell/2)} \leq 2^{m-(\ell-1)/2}$ 
where $\ell$ is the number of colors that have imbalance w.r.t. $\rho$. From the above discussion, we can infer that
  $\prob{\pi}{\rank\left(M_{Y,Z}(f|_\rho)\right) \geq 2^{m- t/1000}} \leq  \prob{\pi}{\ell \leq t/400} \leq \frac{1}{2^{\Omega(t)}}$.
  
\subsubsection*{Part 3: $r$-\simple\ polynomials have low rank.}
\label{sec:p3}
Here we prove that if $f\in \F[X]$ is any $r$-\simple\ polynomial, then for some absolute constant $\delta > 0,$
$\prob{}{\rank(\M_{(Y,Z)}(f|_\rho)) \geq 2^{m - \delta r}} \leq \frac{1}{2^{\Omega(r)}}$.

As $f$ is an $r$-\simple\ polynomial we know that $f = \left(\prod_{i=1}^{r'} L_i\right)\cdot G$, where $r'\leq r$, $L_i$s are linear polynomials, $\forall i \in [r']$ $X_i$ is the set of variables ascribed to $L_i$\, and $X_{r'+1}$ is the set of variables ascribed to $G$. Moreover, $|\cup_{i=1}^{r'} X_i| \geq 400r$. 

To prove the above statement we set up some notation.
Let $f|_\rho = \left(\prod_{i=1}^{r'} L_i|_\rho \right)\cdot G|_\rho$\,. Let $Y_i = \{\rho(x) \mid x \in X_i\} \cap Y$ and $Z_i = \{\rho(x) \mid x \in X_i\} \cap Z$ for each $i \in [r']$. Let $Y' = \cup_{i=1}^{r'} Y_i$ and $Z' = \cup_{i=1}^{r'}Z_i$. Also, let $Y'' = Y \setminus Y'$ and $Z'' = Z \setminus Z'$. Let $U$ denote $\cup_{i=1}^{r'} X_i$ and let $U|_\rho = \cup_{i=1}^{r'} Y_i \cup \cup_{i=1}^{r'} Z_i$. 

In the following claim we show that if $U$ is a large set to begin with then with high probability (over the restriction $\rho$ defined by the sampling algorithm), $U|_\rho$ is also large. 

\begin{claim}
\label{cl:many-variables}
If $|U|\geq 400r$, then $\prob{}{|U|_\rho| \leq 4 r} \leq \frac{1}{2^{\Omega(r)}}.$
\end{claim}

We first finish the proof of Part $3$ of Lemma~\ref{lem:rand-rest} assuming this claim.

We say that a restriction $\rho$ is good if we get $|U|_\rho| \geq 4r$. In what follows we will condition on the event that we have a good $\rho$. 

For a restriction $\rho$, for each $i \in [r']$, we can write $L_i|_\rho(Y_i,Z_i)$ as $L'_i|_\rho(Y_i)+$ $L''_i|_\rho(Z_i)$ as $L_i$s are linear polynomials. Therefore we get 
$\prod_{i=1}^{r'} L_i|_\rho(Y',Z') = \sum_{S \subseteq [r']} ~\prod_{i \in S} L'_i|_\rho(Y_i) \cdot \prod_{j \in [r']\setminus S} L''_j|_\rho(Z_j)$. 

Let $L_S$ denote the polynomial $\prod_{i \in S} L'_i|_\rho(Y_i) \cdot \prod_{j \in [r']\setminus S} L''_j|_\rho(Z_j)$. Note that for all $S \subseteq [r']$, $\rank\left(M_{(Y',Z')} (L_S)\right)$ is at most $1$. Therefore, by the subadditivity of matrix rank, we get that $\rank\left(M_{(Y',Z')} \left(\prod_{i=1}^{r'} L_i|_\rho(Y',Z')\right)\right)  \leq 2^{r'} \leq 2^r\,.$
We can now bound $\rank\left(M_{(Y,Z)}\left( f|_\rho\right)\right)$. 

\begin{align*}
\frac{\rank\left(M_{(Y,Z)}\left( f|_\rho\right)\right)}{2^{{(|Y|+|Z|)}/{2}}} & = \frac{\rank\left(M_{(Y,Z)}\left( \prod_{i=1}^{r'} L_i|_\rho \cdot G|_\rho\right)\right)}{2^{{(|Y|+|Z|)}/{2}}} \\
 & = \frac{\rank\left(M_{(Y',Z')}\left( \prod_{i=1}^{r'} L_i|_\rho \right)\right)}{2^{{(|Y'|+|Z'|)}/{2}}} \cdot \frac{\rank\left(M_{(Y'',Z'')}\left(G|_\rho\right)\right)}{2^{{(|Y''|+|Z''|)}/{2}}}\  \\
& \leq \frac{2^r}{2^{|U|_\rho|/2}} \cdot 1 \leq \frac{2^r}{2^{2r}}  =\frac{1}{2^{r}}. 
\end{align*}
where the second equality follows from Lemma~\ref{lem:mult}. 
Therefore, we have 
$\rank\left(M_{(Y,Z)}\left( f|_\rho\right)\right)  \leq 2^{{(|Y|+|Z|)}/{2}}/2^{r} \leq 2^{m+(1/2)- r}$ for any good $\rho$. 
%
As Claim~\ref{cl:many-variables} tells us that $\rho$ is good with probability $1-1/2^{\Omega(r)}$, we are done. \qed.

Assuming Claim~\ref{cl:many-variables} we are done with the proof of Part 3 of Lemma~\ref{lem:rand-rest}. Given below is the proof of Claim~\ref{cl:many-variables}.

\begin{proof}[Proof of Claim~\ref{cl:many-variables}]
We say that a layer $i \in [d]$ is \emph{touched} by $U$ if there is a variable $x^{(i)}_{u,v} \in U$. We call such an $x^{(i)}_{u,v}$ a \emph{contact edge}. Any layer touched by $U$ has at most $4$ contact edges. As $|U| \geq 400r$, $U$ touches at least $100r$ layers. At least half of the layers will be odd numbered or at least half of them will be even numbered. Let us assume without loss of generality that at least half of them are odd numbered. Let these be $\ell_1, \ell_2, \ldots, \ell_R$, where $R \geq 50r $. Let us fix a contact edge $(u_i,v_i)$ per $\ell_i$ for each $i \in [R]$. Let us denote that these edges by $x^{(\ell_i)}_{(u_i,v_i)}$ for $i \in [R]$.  Let us use an indicator random variable $W_i$ which is set to $1$ if $\rho(x^{(\ell_i)}_{(u_i,v_i)}) \in U|_\rho$ and to $0$ otherwise. Note that $\prob{a,\pi}{W_i =1} = 1/8$, where $a,\pi$ are as in the sampling algorithm. This is because, for odd layers, probability that a fixed edge (among $4$ possible contact edges) is picked by $\pi$  is exactly $1/4$ and for a odd layer $\ell$ the probability that $a_\ell =1$ is exactly $1/2$. Moreover, both these events are independent. Therefore $\avg{}{\sum_{i=1}^R W_i} = R/8 \geq 5r$.  Hence we get, $\prob{}{|U|_\rho| \leq 4r} \leq \prob{}{\sum_{i=1}^R W_i \leq 4r}\leq \frac{1}{2^{\Omega(r)}}$, where the last inequality is by the Chernoff bound.
\end{proof}

\paragraph*{Acknowledgement.} We thank the organizers of the NMI Workshop on Arithmetic Complexity 2016 where this collaboration began. Part of this work was done while SC was affiliated to Chennai Mathematical Institute as a graduate student and SC thanks TCS PhD fellowship.
\bibliography{localref}


\end{document}